\documentclass{article}[12pt]
\usepackage[top=2.5cm, bottom=2.5cm, left=2.5cm, right=2.5cm]{geometry}           
\usepackage{graphicx}
\usepackage{amssymb, amsthm, mathrsfs}
\usepackage{epstopdf}

\usepackage{enumerate}
\usepackage{color}
\usepackage{hyperref}
\usepackage{graphicx}
\usepackage{color}
\usepackage{amsmath, amsthm, slashed}
\usepackage{amssymb}
\usepackage{rotating, bm}
\usepackage{epsfig}
\usepackage{verbatim}
\usepackage{rotating}
\usepackage{graphicx}
\usepackage{tensor}

\newcommand{\pfstep}[1]{\vspace{.5em} {\it \noindent #1.}}

\newcommand{\be}{\begin{equation*}}
\newcommand{\ee}{\end{equation*}}
\newcommand{\ben}[1]{\begin{equation}\label{#1}}
\newcommand{\een}{\end{equation}}
\newcommand{\bea}{\begin{eqnarray}}
\newcommand{\eea}{\end{eqnarray}}
\newcommand{\bean}{\begin{eqnarray*}}
\newcommand{\eean}{\end{eqnarray*}}

\newtheorem{Theorem}{Theorem}

\newtheorem{lemma}{Lemma}

\newtheorem{Proposition}{Proposition}
\newtheorem{Remark}{Remark}
\newtheorem{Test}{Test}

\begin{document}

\title{The initial boundary value problem in General Relativity: the umbilic case}

\author{Grigorios Fournodavlos,\footnote{Sorbonne Universit\'e, CNRS, Universit\'e  de Paris, Laboratoire Jacques-Louis Lions (LJLL), F-75005 Paris, France,
		grigorios.fournodavlos@sorbonne-universite.fr}\;\; Jacques Smulevici\footnote{Sorbonne Universit\'e, CNRS, Universit\'e  de Paris, Laboratoire Jacques-Louis Lions (LJLL), F-75005 Paris, France, jacques.smulevici@sorbonne-universite.fr}}
		
		\date{}

\maketitle
\begin{abstract}
We give a short proof of local well-posedness for the initial boundary value problem in general relativity with sole boundary condition the requirement that the boundary is umbilic. This includes as a special case the totally geodesic boundary condition that we had previously addressed in \cite{FSm2020}. The proof is based on wave coordinates and the key observation that the momentum constraint is always valid for umbilic boundaries. This allows for a greater freedom in the choice of boundary conditions, since imposing the umbilic condition also provides Neumann boundary conditions for three of the four wave coordinates conditions. Moreover, the umbilic condition, being geometric, implies that geometric uniqueness in the sense of Friedrich holds in this specific case. 
\end{abstract}

	\section{Introduction}
	
	In this note, we establish the local well-posedness of the initial boundary value problem (IBVP) for the Einstein vacuum equations
	\begin{align}\label{EVE}
		 \mathrm{Ric}(g)=0,
	\end{align}
	in the specific case of an \emph{umbilic} timelike boundary.

	\subsection{The initial boundary value problem in General Relativity}

	In the standard formulation of the Cauchy problem for the Einstein {}{vacuum} equations, given a Riemannian manifold $(\Sigma, h)$ and a $2$-tensor $k$ satisfying the constraint equations 
	\begin{align}
\label{Hamconst}		\mathrm{R}(h)-|{k}|^2+(\text{tr}{k})^2=&\,0,\\
\label{momconst}		\mathrm{div}{k}-\mathrm{d}\text{tr}{k}=&\,0,
	\end{align}
	where $\mathrm{R}(h)$ is the scalar curvature of the Riemannian metric $h$ and all spatial operators are taken with respect to $h$, the goal is to construct a Lorentzian manifold $(\mathcal{M},g)$ solution to the Einstein equations, together with an embedding of $\Sigma$ into $\mathcal{M}$ such that $(h,{k})$ coincides with the first and second fundamental forms of the embedding. For the IBVP, we now require that $\Sigma$ is a manifold with boundary $\mathcal{S}$. We consider an additional manifold $\mathcal{B}= \mathbb{R} \times \mathcal{S}$, a section $\mathcal{S}_0=\{0\} \times \mathcal{S} $ of $\mathcal{B}$, which is identified with the boundary $\mathcal{S}$ of $\Sigma$ via a diffeomorphim $\psi_{s,s_0}$, and a set of functions $\mathcal{BC}$ on $\mathcal{B}$ representing source terms for the chosen boundary conditions. On top of the constraint equations, the initial and boundary data must also now verify the so-called corner or compatibility conditions, a set of equations involving $h,k, \mathcal{BC}$, their derivatives at all orders on $\mathcal{S}$ and $\mathcal{S}_0$, as well as a given real function $\omega$ on $\mathcal{S}$, which eventually will represent the angle between the initial slice and the timelike boundary.

A solution to the IBVP is then a Lorentzian manifold $(\mathcal{M},{\bf g})$ with a timelike boundary $\mathcal{T}$, an embbeding $\psi_t$ of a neigbhoorhood of $\mathcal{S}_0 \subset \mathcal{B}$ into $\mathcal{T}$, such that the boundary data $\mathcal{BC}$ can be identified with the corresponding data on $\mathcal{T}$, and an embedding $\psi_i$ of $\Sigma$ into $\mathcal{M}$ respecting the {}{initial} data, with $\psi_i(\mathcal{S})=\psi_t(\mathcal{S}_0)$, such that $\psi_t^{-1} \circ [\psi_i]_{|\mathcal{S}}=\psi_{s,s_0}$ and the angle between $\mathcal{T}$ and $\psi_i(\Sigma)$ is $\omega \circ \left[\psi_i^{-1}\right]_{\psi_i(\mathcal{S})}$.
	
There is a priori a large freedom in the choice of boundary conditions. The sources $\mathcal{BC}$ could correspond to the values of tensor fields encoding the geometry of $\mathcal{T}$, for instance, the first or second fundamental forms of $\mathcal{T}$, its conformal geometry, some curvature invariants {}{or} they could correspond to components of geometric tensor fields in some gauge and boundary conditions for the gauge itself. 
	
	
	The IBVP is {}{related} to many important aspects of {}{general relativity} and the Einstein equations, such as numerical relativity, the construction of asymptotically Anti-de-Sitter spacetimes, timelike hypersurfaces emerging as the boundaries of the support of massive matter fields or the study of gravitational waves in a cavity \cite{AKMS} and their nonlinear interactions. This problem was first addressed for the Einstein equations in the seminal work of Friedrich-Nagy \cite{FriedNag}, as well as by Friedrich \cite{Fried95} in the related Anti-de-Sitter setting\footnote{See also \cite{CarVal, EncKam} for extensions and other proofs of well-posedness in the Anti-de-Sitter case.}. Well-posedness {}{of} the IBVP has since been obtained in generalized wave coordinates, see \cite{KRSW} or the recent \cite{AA}\footnote{To be more precise, the boundary data in \cite{AA} relies on an auxiliary wave map equation akin to generalized wave coordinates. This introduces a geometric framework to address the IBVP, albeit for the Einstein equations coupled to the auxiliary wave map equation.} and for various first and second order systems derived from the ADM formulation of the Einstein equations, see for instance \cite{FSm, SarTig} {}{and previous work in numerics \cite{AndYork,FriReu}}. We refer the reader to \cite{SarTig2} for an extensive review of the subject. 
	
	
	\subsection{Geometric uniqueness and completeness} 
	
	One of the {}{remaining} outstanding issues, concerning the study of the Einstein equations in the presence of a timelike boundary, is the \emph{geometric uniqueness} problem of Friedrich {}{\cite{Fried09}}. Apart from the construction of asymptotically Anti-de-Sitter spacetimes {}{\cite{Fried95}}, where the timelike boundary is a conformal boundary at spacelike infinity and our recent work in the totally geodesic case, \cite{FSm2020}, 
all results establishing well-posedness, for some formulations of the {}{IBVP}, impose certain gauge conditions on the boundary, and the boundary data depend on these choices. In particular, given a solution to the Einstein equations with a timelike boundary, different gauge choices will lead to different boundary data, in each of the formulations for which well-posedness is known.

While the totally geodesic case handled in \cite{FSm2020} verifies the geometric uniqueness property, it does not verify \emph{geometric completeness}, in the sense that, it obviously cannot be used to reconstruct all globally hyperbolic vacuum solutions of the Einstein equations with a timelike boundary, since the boundary conditions have no freedom. The umbilic case that we study here is similar, it verifies geometric uniqueness but not geometric completeness. 
	
In the recent \cite{Fried21}, Friedrich proves that the gauge introduced in \cite{FriedNag} verifies geometric completeness. This work also provides an alternative proof to ours in the totally geodesic case. In fact, the strategy of \cite{Fried21} can also be extended to recover the umbilic case that we treat here. Interestingly, the proof of \cite{Fried21}, for the totally geodesic case, relies on the Friedrich-Nagy \cite{FriedNag} gauge and formulation of the IBVP and the a posteriori resolution of a hyperbolic system tangential to the boundary. By comparison, the proof that we provide here relies on the simpler wave gauge and is more direct, but of course it is not clear a priori whether our formulation can be used to study other boundary conditions.


	In the Anti-de-Sitter setting, the geometric uniqueness and completeness problems admits one solution: in \cite{Fried09}, Friedrich proved that one can take the conformal metric of the boundary as boundary data, which is a geometric condition independent of any {}{gauge}. Even in the Anti-de-Sitter setting, it is actually possible to formulate other boundary conditions, such as dissipative boundary conditions, for which one knows how to prove well-posedness, however, with a formulation of the boundary conditions that is gauge dependent and thus, such that we do not know whether geometric uniqueness holds or not. 
	
	\subsection{The initial value problem for umbilic boundary}
	Recall that a hypersurface $\mathcal{T}$ in a Lorentzian manifold $(\mathcal{M},g)$ is said to be \emph{umbilic} if the first and second fundamendal forms $(H, \chi)$ of $\mathcal{T}$ verify
	\begin{equation} \label{def:umbl}
	\chi= \lambda H, 
	\end{equation}
 	for some constant $\lambda \in \mathbb{R}$. We refer to \eqref{def:umbl} as the $\lambda$-umbilic boundary condition and say that $\mathcal{T}$ is $\lambda$-umbilic. 
	Such hypersurfaces occur for instance in black hole spacetimes possessing a photon sphere, such as the Schwarzschild spacetime and the generalizations introduced in \cite{CG}. 
	
	Our main result concerning the IBVP can be formulated as follows. 
	\begin{Theorem}\label{thmA}
		Let $(\Sigma, h, {k})$ be a smooth initial data set for the Einstein {}{vacuum} equations, such that $\Sigma$ is a $3$-manifold with boundary ${}{\partial\Sigma=\mathcal{S}}$. Also, let $\lambda \in \mathbb{R}$.
		For a smooth function $\omega$ defined on $\mathcal{S}$, we assume that the corner conditions (see Section \ref{se:cc}) corresponding to the $\lambda$-umbilic boundary condition hold on $\mathcal{S}$. Then, there exists a smooth Lorentzian manifold $({}{\mathcal{M}}, g)$ solution to the Einstein {}{vacuum} equations with boundary $\partial \mathcal{M}= \widehat{\Sigma} \cup \mathcal{T}$ such that 
		\begin{enumerate}
			\item there exists an embedding $\psi_i$ of $\Sigma$ onto $\widehat{\Sigma}$ with $(h,{k})$ coinciding with the first and second fundamental form of the embedding, \label{th:d}
			\item $\mathcal{T} \cap \widehat{\Sigma}=\psi_i( \mathcal{S})$ and $\mathcal{T}$ is a timelike hypersurface emanating from $i(\mathcal{S})$ at an angle $\omega \circ [\psi_i^{-1}]_{| \psi_i(\mathcal{S})}$ relative to $\hat{\Sigma}$, \label{th:tb}
			\item $\mathcal{T}$ is $\lambda$-umbilic, \label{th:tg}
			\item geometric uniqueness holds: given any other solution $(\mathcal{M}', g')$ verifying \ref{th:d}, \ref{th:tb} and \ref{th:tg}, $(\mathcal{M}, g)$ and $(\mathcal{M}', g')$ are both extensions\footnote{Recall that $(\mathcal{M}, g)$ is an extension of $(\mathcal{M}'', g'')$, if there exists an isometric {}{embedding} ${}{\psi}: \mathcal{M}'' \rightarrow \mathcal{M}$, preserving orientation, such that $\psi \circ i''=i$ and $\psi\circ\psi_t''=\psi_t$, where $i'': \Sigma \rightarrow \mathcal{M}''$, $\psi_t'':\mathcal{B}\rightarrow\mathcal{M}''$ are the embeddings of the initial hypersurface and timelike boundary into $\mathcal{M}''$ respectively.}   of yet another solution verifying \ref{th:d}, \ref{th:tb} and \ref{th:tg}. 
		\end{enumerate}
	\end{Theorem}


%
\subsection{Approach to local well-posedness}

We set up the IBVP with umbilic boundary in a wave coordinates gauge. As we show in Proposition \ref{prop:rs}, under appropriate conditions on the coordinate system $(t,x^A,x)\in [0,T]\times U_A\times [0,\epsilon]$, the reduced IBVP with umbilic boundary for the components of the metric, in a neighborhood of a point on $\mathcal{S}$, takes the form (see Section \ref{subsec:not} for our index notation): 
\begin{align}\label{eq:rs.cond}
\left\{
\begin{array}{lllll}
g^{\alpha\beta}\partial_\alpha\partial_\beta g_{\mu\nu} =Q_{\mu\nu}(\partial g,\partial g),\vspace*{.1in}\\
g_{xt}=g_{xA}=0,\quad &\text{on $\mathcal{T}$},\\
\partial_xg_{ij}=2\lambda (g_{xx})^{\frac{1}{2}}g_{ij},\quad \partial_xg_{xx}=6\lambda(g_{xx})^\frac{3}{2},\quad&\text{on $\mathcal{T}$},\vspace*{.1in}\\
g_{00}=-\Phi^2,\quad g_{0p}=v_p,\quad g_{pq}=h_{pq},\quad \partial_0g_{pq}=-2\Phi k_{pq}+\mathcal{L}_vg_{ij},\quad&\text{on $\Sigma$},\\
\partial_0g_{00}= 2\Phi^3 \mathrm{tr} k-4v^p\Phi\partial_p\Phi+\Phi^2\partial^pv_p-\Phi^2g^{pq}\mathcal{L}_vg_{pq},\quad&\text{on $\Sigma$},\\
\partial_0g_{0p}=\frac{1}{2}\Phi^2h^{rs}(\partial_rh_{sp}+\partial_sh_{rp}-\partial_ph_{rs})-\Phi\partial_p\Phi+2v^q(\mathcal{L}_vg_{pq}-2\Phi k_{pq}-\partial_pv_q),\quad&\text{on $\Sigma$},
\end{array}
\right.
\end{align}
where $\Phi,v^p$ are the initial lapse and shift vector field, with $n=\Phi^{-1}(\partial_t-v)$ being the future unit normal to $\Sigma$, and $\mathcal{L}$ the Lie derivative operator. 
The initial conditions for the time derivatives of $g_{00},g_{0p}$ are such that the wave coordinates condition $\Gamma_\mu=0$ is valid on $\Sigma$. The first set of boundary conditions in \eqref{eq:rs.cond} is gauge, induced by our choice of coordinates (see Section \ref{subsec:coord}), while the second set of {\it Robin} type boundary conditions for $g_{ij},g_{xx}$ is induced by the $\lambda$-umbilic boundary condition \eqref{def:umbl}.
\begin{Remark}\label{rem:lapse}
In order for the initial configurations to be consistent with the $\lambda$-umbilic boundary condition \eqref{def:umbl} and the compatibility conditions (see Section \ref{se:cc}), we cannot choose the coordinates to be initially Gaussian, but we must instead allow for non-trivial initial lapse and shift vector field $\Phi,v^p:U_A\times[0,\epsilon]$. Indeed, in view of the boundary conditions, $\Phi,v_p=g_{pq}v^q$ must verify
\begin{align*}
\partial_xg_{00}=&\,2\lambda (g_{xx})^\frac{1}{2}g_{00}\qquad\Leftrightarrow\qquad \partial_x\Phi = 2\lambda(g_{xx})^\frac{1}{2}\Phi \qquad\text{on $\mathcal{S}$},\\
\partial_xg_{0p}=&\,2\lambda (g_{xx})^\frac{1}{2}g_{0p}\qquad\Leftrightarrow\qquad \partial_xv_p= 2\lambda(g_{xx})^\frac{1}{2}v_p \qquad\text{on $\mathcal{S}$}.
\end{align*}
Hence, after prescribing their values on $\mathcal{S}$, the boundary conditions and the higher order compatibility conditions of Section \ref{se:cc} imply that the derivatives of $\Phi,v_p$ are fixed at all orders on $\mathcal{S}$, if one wants to construct smooth solutions. Note that the value of $g_{03}=g_{xt}$ on $\mathcal{S}$ is forcibly zero so that the first boundary condition in \eqref{eq:rs.cond} is initially verified. 
In the totally geodesic case, $\lambda=0$, the Gaussian choice $\Phi=1,v^p=0$ on $\Sigma$ is permitted. 
\end{Remark}
\begin{Remark}
The IBVP \eqref{eq:rs.cond} contains $10$ boundary conditions for the $10$ metric components $g_{xi}, g_{ij},g_{xx}$, $i,j=0,1,2$, and we are a priori not allowed to impose additional conditions. In some previous works concerning the IBVP in wave coordinates, the wave coordinates condition $\Gamma_i=0$ was part of the boundary conditions, such that they could be propagated by the equations (see \cite{KRSW}). In our case, we do not impose $\Gamma_i=0$ on the boundary, instead, the boundary conditions for the propagation of the gauge are derived from the umbilic condition, see Section \ref{sec:gauge.prop}. 
\end{Remark}
Once we have formulated the IBVP for the Einstein vacuum equations with umbilic boundary as the reduced system \eqref{eq:rs.cond}, local well-posedness follows by standard theory and the domain of dependence property. Indeed, it amounts to solving \eqref{eq:rs.cond} locally near a point on the boundary $\mathcal{S}$. This is based on deriving standard energy estimates. The only part which requires some attention is how to treat the arising boundary integrals for the components $g_{ij},g_{xx}$ that satisfy Robin type boundary conditions. We recall how to derive such energy estimates for the wave equation with boundary in the appendix.

\subsection{Overview} 

In Section \ref{se:setup}, we present the set-up for the initial boundary problem with umbilic boundary. In particular, we introduce the wave coordinates that we will work with and the boundary conditions for the reduced system. In Section \ref{se:cc}, we review the derivation of the corner conditions for an umbilic boundary. Section \ref{sec:gauge.prop} is devoted to the recovery of the Einstein equations, and in particular, the derivation of the boundary conditions for the gauge from the umbilic boundary condition. The key observation here is that the momentum constraint for the umbilic timelike boundary $\mathcal{T}$ is automatically satisfied, see Lemma \ref{lem:momconst}. Finally, in Appendix \ref{ap:werbc}, we briefly review the energy estimates for the wave equation with Robin boundary conditions which lead to the well-posedness of the IBVP for the reduced system \eqref{eq:rs.cond}. 

After recovering the full Einstein vacuum equations \eqref{EVE} and the umbilic condition \eqref{def:umbl} from the solution to \eqref{eq:rs.cond}, see Propositions \ref{prop:rs}, \ref{prop:gauge}, the geometric uniqueness statement in Theorem \ref{thmA} follows in a straightforward way from the homogeneity of the umbilic condition \eqref{def:umbl} and its geometric nature, thus, completing the proof of Theorem \ref{thmA}. 

\begin{quote}
\textbf{Acknowledgements.} {\small We would like to thank M.T. Anderson for many interesting discussions on the subject, in particular, for forwarding us E.~Witten's question on the IBVP with umbilic boundary, which led us to write this paper. Both authors are supported by the \texttt{ERC grant 714408 GEOWAKI}, under the European Union's Horizon 2020 research and innovation program.}
\end{quote}

\subsection{Notations}\label{subsec:not}

We use Einstein's summation for repeated upper and lower indices.
Greek indices $\alpha,\beta,\gamma,\mu,\nu$ range over $\{0,1,2,3\}$ and refer to the spacetime coordinates $(t,x^A,x)$, $A=1,2$, with $x^0=t,x^3=x$. Latin indices $i,j,l,m$ range over $\{0,1,2\}$ and refer to the coordinates $(t,x^A)$. These are used for the components of tensors that are tangential to the timelike boundary $\mathcal{T}$. Latin indices $p,q,r$ range over $\{1,2,3\}$ and refer to the coordinates $(x^A,x)$ or $(y^A,y)$, which are also used for the components of tensors tangential to the initial hypersurface $\Sigma$.

\section{Setting up the IBVP with umbilic boundary in wave coordinates} \label{se:setup}
Let $(\mathcal{M},g)$ be an $3+1$ dimensional Lorentzian manifold with boundary $\mathcal{T}\cup\Sigma$, where $\Sigma$ is a Cauchy hypersurface in the sense of \cite{HFS} and $\mathcal{T}$ is a timelike boundary with outgoing unit normal $N$. We assume that the boundary $\mathcal{T}$ is umibilic, i.e. the second fundamental form $\chi$ of $\mathcal{T}$ verifies \eqref{def:umbl}, for some constant $\lambda\in \mathbb{R}$. In particular, the mean curvature of $\mathcal{T}$ is constant 
\begin{align*}
\mathrm{tr} \chi= 3 \lambda,
\end{align*}
where the trace is taken relative to the induced metric $H$ on $\mathcal{T}$. 
We assume for simplicity that $\Sigma$ is orthogonal to $\mathcal{T}$. The general case where $\Sigma$ intersects $\mathcal{T}$ at an arbitrary angle $\omega$ can be reduced to the orthogonal one by first solving the standard initial value problem in a domain of dependence region and then choosing a new Cauchy slice orthogonal to the boundary, see \cite[Section 3.2]{FSm2020} for an example of such a reduction.

\subsection{The choice of coordinates}\label{subsec:coord}

Recall that on a globally hyperbolic manifold with boundary (see \cite{HFS}), the wave equation 
$$
\square_g \psi=0,
$$
with initial data on $\Sigma$ and either Dirichlet or Neumann boundary conditions on $\mathcal{T}$ is well-posed. This follows from standard energy estimates and the domain of dependence property. 

Let $y^A,y$, $A=1, 2$ be local coordinates on the Cauchy hypersurface $\Sigma$ defined in a neighborhood of a point on $\partial \Sigma$, such that $y$ is a boundary defining function of $\partial \Sigma$:
\begin{align*}
y \ge 0,\qquad \{y=0\} =\partial \Sigma.
\end{align*}
Also, let $\epsilon > 0$ be small enough such that the range of $(y^A,y)$ contains a set $U_A\times [0, \epsilon]$, where $U_A$ is some open set of $\mathbb{R}^2$. 

In $[0,T]\times U_A\times [0, \epsilon]\subset\mathcal{M}$, we consider coordinates $(x^\alpha)=(t,x^A,x)$, $A=1,2$, such that 

\begin{itemize}
	\item The wave coordinates condition is verified 
\begin{align*}	
	\square_g x^\gamma=0\qquad\Leftrightarrow\qquad g^{\alpha\beta}\Gamma_{\alpha\beta}^\gamma=0.
	\end{align*}
	\item $t,x^A$ verify the homogogeneous Neumann boundary conditions $N(t)=N(x^A)=0$ on $\mathcal{T}$. 
	\item $x$ verifies the Dirichlet boundary condition $x=0$ on $\mathcal{T}$. 
	\item On the initial hypersurface $\Sigma$, $t=0$ and $(x^A,x)=(y^A,y)$. 
	\item With respect to the future normal $n$ to $\Sigma$, $n(t)=\Phi^{-1}$ and $n(x^A)=-\Phi^{-1}v^A$, and $n(x)=-\Phi^{-1}v^3$ on $\Sigma$, where $\Phi,v^p:U_A\times [0, \epsilon]\to\mathbb{R}$ are functions of our choice, consistent with $\mathcal{T}$ being $\lambda$-umbilic, see Remark \ref{rem:lapse}. 
\end{itemize}

Each coordinate function $t, x^A,x$, $A=1,2$ verifies a wave equation with either homogeneous Dirichlet or Neumann boundary data and so they are well-defined. Moreover, at least in a sufficiently small spacetime neighborhood of the corner $\partial\Sigma$, $x$ is a boundary definining function of $\mathcal{T}$, and $(t,x^A,x)$ forms a local coordinate system. 

\subsection{The reduced equations with the induced initial and boundary conditions}

Recall that for any Lorentzian manifold $(\mathcal{M},g)$, 
\begin{align}\label{Ric.exp}
\mathrm{Ric}_{\mu \nu}(g)=- \frac{1}{2}g^{\alpha\beta}\partial_\alpha\partial_\beta g_{\mu \nu} +\frac{1}{2} Q_{\mu \nu} (\partial g, \partial g)+ \frac{1}{2}\left( \partial_\mu \Gamma_\nu + \partial_\nu \Gamma_\mu \right), 
\end{align}
where $Q_{\mu \nu}(\partial g, \partial g)$ is a quadratic form in $\partial g$, with coefficients that are rational functions of $g$, and  
\begin{equation}\label{eq:gu}
\Gamma_{\mu}=g_{\mu \gamma} g^{\alpha \beta} \Gamma^\gamma_{\alpha \beta}. 
\end{equation}
\begin{Proposition}\label{prop:rs}
Consider the region $[0,T] \times [0, \epsilon] \times U_A\subset\mathcal{M}$, covered by the wave coordinates $(t,x,x^A)$, for some $T >0$, as defined in Section \ref{subsec:coord}. Then the IBVP for the Einstein vacuum equations with $\lambda$-umbilic boundary reduces to the system of equations and initial/boundary conditions \eqref{eq:rs.cond}. 

Moreover, given a solution to \eqref{eq:rs.cond}, it follows that $\{x=0\}=\mathcal{T}$ is $\lambda$-umbilic and the wave coordinates condition for the coordinate $x$, $\Gamma_x=0$, is verified on the boundary. 
\end{Proposition}
\begin{proof}
In view of \eqref{Ric.exp} and our choice of coordinates in Section \ref{subsec:coord}, the Einstein vacuum equations \eqref{EVE} reduce to the system of wave equations:
\begin{eqnarray}\label{eq:rs}
g^{\alpha\beta}\partial_\alpha\partial_\beta g_{\mu \nu}=  Q_{\mu \nu}(\partial g, \partial g).
\end{eqnarray}

\pfstep{\bf Initial data for the metric} From our choice of initial data for the wave coordinates, the following initial conditions hold
\begin{align}\label{init.metric}
g_{00}=-\Phi^2, \qquad g_{0p}=v_p, \qquad
g_{pq}=h_{pq}, \qquad \partial_t g_{pq}= -2\Phi k_{pq}+\mathcal{L}_vg_{pq},\qquad\text{on $\Sigma$},
\end{align}
where $h_{pq}$, $k_{pq}$ are the components of the first and second fundamental form of $\Sigma$ in the $(y^A,y)$ coordinate system. 

The initial conditions for $\partial_t g_{0\mu}$ are derived, as usual, from the fact that the wave coordinates condition $\Gamma_\mu=0$ is verified on $\Sigma$.  This gives
\begin{align}\label{init.metric.2}
\begin{split}
	\partial_tg_{00}=&\, 2\Phi^3 \mathrm{tr} k-4v^p\Phi\partial_p\Phi+\Phi^2\partial^pv_p-\Phi^2g^{pq}\mathcal{L}_vg_{pq} \qquad\Leftrightarrow\qquad \Gamma_0=0,\qquad\text{on $\Sigma$},\\ 
	\partial_tg_{0p}=&\, \frac{1}{2}\Phi^2h^{rs}(\partial_rh_{sp}+\partial_sh_{rp}-\partial_ph_{rs})-\Phi\partial_p\Phi+2v^q(\mathcal{L}_vg_{pq}-2\Phi k_{pq}-\partial_pv_q)\;\Leftrightarrow\; \Gamma_p=0,\;\text{on $\Sigma$}.
	\end{split}
	\end{align}
	
\pfstep{\bf Boundary conditions for the shift vector field} From the definition of the coordinates in Section \ref{subsec:coord}, we have the boundary conditions
\begin{equation}\label{gxi.bdcond}
g_{xt}=0, \qquad g_{xA}=0,\qquad\text{on $\mathcal{T}$},
\end{equation}
which amounts to saying that the shift vector field $g^{xi}\partial_i$ of the constant $x$-hypersurfaces is normal to the boundary $\mathcal{T}=\{x=0\}$.

\pfstep{\bf Boundary conditions for umbilic boundary}
The induced metric on the boundary $\mathcal{T}$ takes the form
\begin{align*}
H= g_{tt}dt^2+ 2 g_{t A} dx^A dt + g_{AB}dx^A dx^B,
\end{align*}
and the components of the second fundamental form are given by
\begin{align*}
\chi_{ij}= g(D_{\partial_i} \partial_j , N), 
\end{align*}
where $N=-(g_{xx})^{-\frac{1}{2}}\partial_x$ is the outgoing unit normal to $\mathcal{T}$. 

At $x=0$, we have 
\begin{align*}
\chi_{ij}=-\frac{1}{2}(g_{xx})^{-\frac{1}{2}}(\partial_ig_{jx}+\partial_jg_{ix}-\partial_xg_{ij})=\frac{1}{2}(g_{xx})^{-\frac{1}{2}}\partial_xg_{ij}.
\end{align*}
Hence, the $\lambda$-umbilic condition \eqref{def:umbl} implies the Robin type boundary conditions
\begin{align}\label{gij.bdcond}
\partial_xg_{ij} = 2\lambda (g_{xx})^{\frac{1}{2}} g_{ij},\qquad\text{on $\mathcal{T}$}.
\end{align}
In turn, the wave coordinates condition  $\Gamma_x=0$, evaluated along the boundary, gives 
\begin{align*}
0=g_{x\gamma}g^{ij} \Gamma_{ij}^\gamma + g_{x\gamma}g^{xx} \Gamma_{xx}^\gamma=-\frac{1}{2} g^{ij}\partial_xg_{ij}+\frac{1}{2}g^{xx}\partial_xg_{xx}=-3\lambda (g_{xx})^\frac{1}{2}+\frac{1}{2}g^{xx}\partial_xg_{xx}.
\end{align*}
Hence, we obtain another Robin type boundary condition for $g_{xx}$:
\begin{align}\label{gxx.bdcond}
\partial_xg_{xx}= 6 \lambda (g_{xx})^{\frac{3}{2}}, \qquad\text{on $\mathcal{T}$}.
\end{align}

Notice that the derivations for \eqref{gxi.bdcond}, \eqref{gxx.bdcond} are reversible, for a solution to \eqref{eq:rs.cond}. Therefore, the last part of the proposition is also confirmed.
\end{proof}

\section{Higher order compatibility conditions with umbilic boundary} \label{se:cc}

Apart from the boundary conditions \eqref{gxi.bdcond}, \eqref{gij.bdcond}, \eqref{gxx.bdcond}, the initial data \eqref{init.metric}, \eqref{init.metric.2} must also satisfy compatibility conditions at all higher orders on the boundary of the initial hypersurface, $\partial\Sigma=\mathcal{S}$. These are found by iteratively applying the differential operator \eqref{eq:rs} to the boundary conditions and evaluating the resulting expression at $x=0$. We first illustrate the procedure in an abstract way. Let $\mathcal{BC}$ be the boundary condition operator for $g_{\mu\nu}$, ie. $\mathcal{BC}g_{\mu\nu}=0$. Then we compute 
\begin{align*}
g^{\alpha\beta}\partial_\alpha\partial_\beta(\mathcal{BC}g_{\mu\nu})=[g^{\alpha\beta}\partial_\alpha\partial_\beta,\mathcal{BC}]g_{\mu\nu}+\mathcal{BC}(g^{\alpha\beta}\partial_\alpha\partial_\beta g_{\mu\nu})=[g^{\alpha\beta}\partial_\alpha\partial_\beta,\mathcal{BC}]g_{\mu\nu}+\mathcal{BC}[Q_{\mu\nu}(\partial g,\partial g)]
\end{align*}
On the other hand, expanding the wave operator in the LHS, evaluating the resulting expression at $x=0$, and using the boundary condition for $g_{\mu\nu}$ gives
\begin{align*}
g^{\alpha\beta}\partial_\alpha\partial_\beta(\mathcal{BC}g_{\mu\nu})\overset{\eqref{gxi.bdcond}}{=}g^{xx}\partial_x\partial_x(\mathcal{BC}g_{\mu\nu})+g^{ij}\partial_i\partial_j(\mathcal{BC}g_{\mu\nu})=g^{xx}\partial_x\partial_x(\mathcal{BC}g_{\mu\nu}),\qquad\text{at $x=0$}, 
\end{align*}
since $\partial_i$, $\partial_j$ are tangential to the boundary. 
Hence, the next order compatibility condition for the initial datum of $g_{\mu\nu}$ reads:
\begin{align*}
\partial_x\partial_x(\mathcal{BC}g_{\mu\nu})=g_{xx}[g^{\alpha\beta}\partial_\alpha\partial_\beta,\mathcal{BC}]g_{\mu\nu}+g_{xx}\mathcal{BC}[Q_{\mu\nu}(\partial g,\partial g)],\qquad\text{at $x=0$}.
\end{align*}
Setting $\mathcal{BC}^{(n)}g_{\mu\nu}=0$ to be the compatibility condition of $g_{\mu\nu}$ of order $n$, $\mathcal{BC}^{(1)}:=\mathcal{BC}$, and repeating the above procedure, we obtain
\begin{align}\label{gmunu.comp.n}
\mathcal{BC}^{(n+1)}g_{\mu\nu}:=\partial_x\partial_x(\mathcal{BC}^{(n)}g_{\mu\nu})-H_{xx}[g^{\alpha\beta}\partial_\alpha\partial_\beta,\mathcal{BC}^{(n)}]g_{\mu\nu}-H_{xx}\mathcal{BC}^{(n)}[Q_{\mu\nu}(\partial g,\partial g)]=0,
\end{align}
as the $n+1$ order compatibility condition at $x=0$.

Note that we can compute all derivatives of $g_{\mu\nu}$ on $\Sigma$, in terms of the initial conditions \eqref{init.metric}, \eqref{init.metric.2}, using the reduced equations \eqref{eq:rs}, so the above conditions are well-defined at the level of the initial data.

As an example, the boundary condition \eqref{gxi.bdcond} for the metric components $g_{xi}$ results in the following compatibility conditions at $x=0$:
\begin{align*}
\partial_x\partial_x g_{xi}=&\,g_{xx}Q_{xi}(\partial g,\partial g),\\
\partial_x\partial_x\partial_x\partial_xg_{xi}=&-H_{xx}g^{lm}\partial_l\partial_m[g_{xx}Q_{xi}(\partial g,\partial g)]+g_{xx}[g^{\alpha\beta}\partial_\alpha\partial_\beta,\partial_x\partial_x]g_{xi}+g_{xx}\partial_x\partial_x[Q_{xi}(\partial g,\partial g)]\\
&\text{etc.}
\end{align*}
%
For the Robin type boundary conditions \eqref{gij.bdcond}, \eqref{gxx.bdcond}, at $x=0$, we have
\begin{align*}
\partial_x\partial_x\partial_x  g_{ij}=&\,\partial_x\partial_x(2\lambda g_{xx}^\frac{1}{2}g_{ij})-H_{xx}\partial_x(g^{\alpha\beta})\partial_\alpha\partial_\beta g_{ij}+g_{xx}\partial_x[Q_{ij}(\partial g,\partial g)]-2\lambda g_{xx} g^{\alpha\beta}\partial_\alpha\partial_\beta(g_{xx}^\frac{1}{2}g_{ij})\\
\partial_x\partial_x\partial_x  g_{xx}=&\,\partial_x\partial_x(6\lambda g_{xx}^\frac{3}{2})-H_{xx}\partial_x(g^{\alpha\beta})\partial_\alpha\partial_\beta g_{xx}+g_{xx}\partial_x[Q_{xx}(\partial g,\partial g)]-6\lambda g_{xx} g^{\alpha\beta}\partial_\alpha\partial_\beta(g_{xx}^\frac{3}{2})\\
&\text{etc.}
\end{align*}
The above RHSs could be further simplified by plugging in \eqref{eq:rs} whenever the wave operator $g^{\alpha\beta}\partial_\alpha\partial_\beta$ acts on a single metric component, after applying chain rule to each term. 

%
\subsection{Compatibility conditions at the level of the geometric initial data}

The higher order compatibility conditions \eqref{gmunu.comp.n}, for the reduced IBVP with umbilic boundary \eqref{eq:rs.cond}, translate to compatibility conditions for $h,k$ on $\mathcal{S}$, by using the form of the initial configurations \eqref{init.metric}, \eqref{init.metric.2}. According to Remark \ref{rem:lapse}, purely spatial derivatives of $g_{0\beta}$ can be written as functions of $h_{xx}$ and its derivatives. If additionally a time derivative acts on $g_{0\beta}$, then the corresponding term can be expressed in terms of $h_{xx},k_{pq}$ and their derivatives.
This way \eqref{gmunu.comp.n} can be rewritten purely in terms of the components of the geometric initial data for the Einstein vacuum equations, and their derivatives at all orders, expressed in the coordinate system $(x^A,x)$ or $(y^A,y)$.

\section{Propagation of the gauge}\label{sec:gauge.prop}

Consider a solution to the reduced IBVP \eqref{eq:rs.cond}. According to Proposition \ref{prop:rs}, $\{x=0\}=\mathcal{T}$ is $\lambda$-umbilic and the wave coordinates condition for $x$ is verified on the boundary:
\begin{align}\label{wave.x}
\Gamma_x=0,\qquad\text{on $\mathcal{T}$}.
\end{align}

\subsection{The momentum constraint}

The following well-known fact, which follows directly from the umbilic condition \eqref{def:umbl}, is key to the propagation of the wave gauge. Indeed, it holds irrespectively of $(\mathcal{M},g)$ verifying the Einstein vacuum equations, and thus, it is applicable to the solutions obtained from the reduced IBVP \eqref{eq:rs.cond}. 
\begin{lemma} \label{lem:momconst}
The momentum constraint 
\begin{align*}
\partial_i\mathrm{tr}\chi -\nabla^j\chi_{ji}=0\qquad\Leftrightarrow\qquad \mathrm{Ric}_{\alpha i}(g)N^\alpha=0,
\end{align*}
is automatically satisfied for an umbilic boundary. 
\end{lemma}
\begin{proof}
It is a direct consequence of the twice contracted Codazzi equations, restricted to the boundary:
\begin{align*}
\mathrm{Ric}_{\alpha i}(g)N^\alpha= \partial_i\mathrm{tr}\chi -\nabla^j\chi_{ji},\qquad\text{on $\mathcal{T}$},
\end{align*}
where $\nabla$ is the Levi-Civita connection of $H$. Since $\chi=\lambda H$, for some $\lambda\in \mathbb{R}$, $\text{tr}\chi$ is constant and $\nabla$ annihilates $\chi$, confirming the conclusion. 
\end{proof}

\subsection{An IBVP for $\Gamma_\mu$ with homogeneous Dirichlet/Neumann conditions}

In order to prove that a solution to the reduced IBVP \eqref{eq:rs.cond}, we need to propagate the wave gauge. For the initial value problem, this is done by observing that the components $\Gamma_\mu$, given by \eqref{eq:gu}, satisfy a system of wave equations with trivial initial data. For the IBVP, to propagate the vanishing of $\Gamma_\mu$, homogeneous boundary conditions are also required. Since we are not allowed to impose any further boundary conditions, these will have to be induced from the umbilic condition \eqref{def:umbl}.
\begin{Proposition}\label{prop:gauge}
Let $g$ be a solution to the reduced IBVP \eqref{eq:rs.cond} in $[0,T]\times U_A\times[0,\epsilon]$. Then the components $\Gamma_\mu$ satisfy the IBVP
\begin{align}\label{eq:Gamma.mu}
\left\{
\begin{array}{lll}
g^{\alpha\beta}\partial_\alpha\partial_\beta \Gamma_\mu + L^{\alpha\beta}_{\mu}(\partial g)\partial_\alpha\Gamma_\beta=0,\quad &\text{in $[0,T]\times U_A\times [0,\epsilon]$},\vspace*{.1in}\\
\Gamma_x=\partial_x\Gamma_i=0,\quad &\text{on $\mathcal{T}$},\vspace*{.1in}\\
\Gamma_\mu=\partial_0\Gamma_\mu=0,\quad&\text{on $\Sigma$},
\end{array}
\right.
\end{align}
where $L^{\alpha\beta}_{\mu}(\partial g)$ is linear in $\partial g$. 
In particular, the $\Gamma_\mu$'s vanish in a domain of dependence region.   
\end{Proposition}
\begin{proof}
Given a solution to the reduced equations \eqref{eq:rs}, from the formula \eqref{Ric.exp} of the Ricci tensor, we have
\begin{eqnarray}
\mathrm{Ric}_{\mu\nu}(g)=\partial_\mu \Gamma_\nu + \partial_\nu \Gamma_\mu, \label{eq:rt}
\end{eqnarray}
where $\Gamma_\mu$ is defined via \eqref{eq:gu}. 
Taking the divergence of \eqref{eq:rt} with respect to $\mu$ and using the twice contracted second Bianchi idendity, we obtain a linear system of wave equations for $\Gamma_\mu$ as in \eqref{eq:Gamma.mu}.

By virtue of the initial conditions \eqref{init.metric.2}, and the validity of the constraint equations \eqref{Hamconst}-\eqref{momconst} initially, it holds $\Gamma_\mu=\partial_t\Gamma_\mu=0$ on $\Sigma$. 

On the other hand, in view of the umbilic condition and Lemma \ref{lem:momconst}, $\mathrm{Ric}_{xi}(g)=0$ on the boundary, since $N=-(g_{xx})^{-\frac{1}{2}}\partial_x$. Thus, using the validity of \eqref{wave.x}, the identity \eqref{eq:rt} with indices $x$ and $i$ gives the homogeneous Neumann boundary conditions:
\begin{align*}
\partial_x \Gamma_i = 0,\qquad\text{on $\mathcal{T}$},
\end{align*}
as asserted.
\end{proof}

\appendix

\section{The wave equation with Robin boundary conditions} \label{ap:werbc}

Consider the following IBVP for the linear wave equation
\begin{eqnarray} \label{eq:tp}
\left\{
\begin{array}{lll}
\square_g \phi&=&0, \quad \\
 N(\phi)&=& R \phi,  \quad \mathrm{on\,\,} \mathcal{T},  \\
\phi &=& f_0, \quad n(\phi)= f_1, \quad  \mathrm{on\,\,} \Sigma.
\end{array}
\right.
\end{eqnarray}
on a globally hyperbolic manifold $(\mathcal{M},g)$, having a Cauchy hypersurface $\Sigma$ with future unit normal $n=\Phi^{-1}(\partial_t-v)$ and a timelike boundary $\mathcal{T}$ with outgoing unit normal $N$. 

In the application to \eqref{eq:rs.cond}, the coefficient $R$ appearing in the boundary conditions is a given power of $g_{xx}$, as in \eqref{gij.bdcond}, \eqref{gxx.bdcond}, since $N=-(g_{xx})^{-\frac{1}{2}}\partial_x$. Note that the covariant wave operator $\square_g=g^{\alpha\beta}\partial_\alpha\partial_\beta-g^{\alpha\beta}\Gamma_{\alpha\beta}^\gamma\partial_\gamma$ does not change the form of the reduced equations \eqref{eq:rs}. We merely consider it here for convenience in employing the divergence theorem below.

The local well-posedness of the IBVP \eqref{eq:tp} is based on a standard energy estimate that we here briefly recall.

\begin{Proposition}
Assume that $(\mathcal{M},g)$ is a Lorentzian manifold with timelike boundary $\mathcal{T}$ and consider a domain of dependence region $\mathcal{D}$ containing a neighborhood of a point on the boundary, foliated by Cauchy hypersurfaces $\Sigma_t$, with $\Sigma_0= \Sigma$. Let $(t,x^A,x)$ be a coordinate chart that covers $\mathcal{D}$, such that $x$ is a boundary defining function of $\mathcal{T}$. 
Also, let $\phi$ be a solution to \eqref{eq:tp}.
Then, for any $t \ge 0$, we have 
$$
\int_{\Sigma_t} \left(|\partial \phi|^2+ \phi^2\right)  \mathrm{vol}_{\Sigma_t} \le C_t \int_{\Sigma} \left(|\partial \phi |^2 + \phi^2\right) \mathrm{vol}_{\Sigma}.
$$
where the constant $C_t$ depends on the $C^1(\mathcal{D})$ norms of $g,R$. 
\end{Proposition}
\begin{proof}
Consider the energy momentum tensor 
$$
T_{\alpha \beta}[\phi]= \partial_\alpha \phi \partial_\beta \phi - \frac{1}{2}g_{\alpha \beta} g^{\mu\nu} \partial_\mu \phi \partial_\nu\phi.
$$
Let $T_{0\beta}[\phi]=T[\phi](\partial_t, \partial_\beta)$ and define the energy at time $t$ as 
\begin{eqnarray} \label{def:energyrob}
E(t):=\int_{\Sigma_t} T[\phi](\partial_t, n_{\Sigma_t})\mathrm{vol}_{\Sigma_t}  + C_R \int_{\Sigma_t} \phi^2\mathrm{vol}_{\Sigma_t} , 
\end{eqnarray}
for some constant $C_R$ depending on the $C^0(\mathcal{D})$ norms of $g,R$, where $n_{\Sigma_t}$ is the future unit normal to $\Sigma_t$.

A standard application of the divergence theorem to the vector field $T_0^\alpha[\phi]=g^{\alpha\beta}T_{0\beta}[\phi]$, in the domain of dependence region considered, allows to estimate $E(t)$ by its time integral, the initial data, and the boundary terms: 
\begin{align*}
E(t)\leq&\, E(0)+C\int^t_0E(s)ds+\int_{x=0, \,0 \le s \le t} \sqrt{-H} T_{0\beta} [\phi] N^\beta ds dx^A\\
=&\,E(0)+C\int^t_0E(s)ds+\int_{x=0, \,0 \le s \le t} \sqrt{-H} \partial_t \phi  R \phi ds dx^A,
\end{align*}
where $\sqrt{-H}$ denotes the induced volume form on the boundary and where we have discarded the flux term from the null boundary of the domain of dependence region, since it always has a favorable sign.

The last term above can be absorbed in the energy by integrating by parts in $\partial_t$ and then using trace inequality:
\begin{eqnarray*}
\int_{x=0,\, 0 \le s \le t} \sqrt{-H} \partial_t \phi  R\phi ds dx^A &=&- \int_{x=0,\, 0 \le s \le t}  \frac{1}{2}\partial_t(\sqrt{-H}  R ) \phi^2 ds dx^A\\
&&+ \int_{\{x=0\} \cap \Sigma_s}  \frac{1}{2}\sqrt{-H} R \phi^2 dx^A\bigg|^{s=t}_{s=0}\\
&\leq& C\int_0^t E(s)ds +D_R \int_{\Sigma_t} \phi^2\mathrm{vol}_{\Sigma_t}  +C E(0),
\end{eqnarray*}
Choosing $C_R$ large enough in \eqref{def:energyrob} to begin with, we deduce that
$$
E(t) \leq CE(0)+C \int_0^t E(s)ds. 
$$
Thus, the desired energy estimate follows by Gr\"onwall's inequality. 
\end{proof}

\end{document}